\documentclass[a4paper]{article}

\usepackage{amsmath,amsfonts,amssymb,amsthm}
\usepackage{graphicx}
\usepackage{color}
\usepackage{mathrsfs} 
\usepackage{float}
\usepackage[version=3]{mhchem}

\DeclareMathOperator{\supp}{supp}

\newcommand{\RR}{{\mathbb R}}

\newcommand{\mal}{\circ}
\newcommand{\dd}[2]{\frac{\text{d} #1}{\text{d} #2}}

\newcommand{\xx}{\text{ex}}
\newcommand{\tot}{\text{tot}}

\newcommand{\FC}{C}
\newcommand{\ka}{\kappa}
\newcommand{\rr}{r}
\newcommand{\ro}{{r_\ast}}

\newtheorem{thm}{Theorem}
\newtheorem{pro}[thm]{Proposition}
\newtheorem{lem}[thm]{Lemma}
\newtheorem{cor}[thm]{Corollary}

\newcommand\col[1]{#1}

\begin{document}


\thispagestyle{empty}

\begin{center}
\Large
Enzyme allocation problems in kinetic metabolic networks:
Optimal solutions are elementary flux modes

\vspace{2ex}
\large
Stefan M\"uller$^{1,2}$, Georg Regensburger$^1$, Ralf Steuer$^3$

\vspace{2ex}
\normalsize
$^1$%
Johann Radon Institute for Computational and Applied Mathematics,\\
Austrian Academy of Sciences\\
Altenbergerstra{\ss}e 69, 4040 Linz, Austria

$^{2}$%
CzechGlobe,\\
Academy of Sciences of the Czech Republic\\
B\v{e}lidla 986/4a, 603 00 Brno, Czech Republic

$^{3}$%
Institute for Theoretical Biology,\\
Humboldt University Berlin\\
Invalidenstra{\ss}e 43, 10115 Berlin, Germany

\vspace{2ex}
\tt
stefan.mueller@ricam.oeaw.ac.at
\rm
\end{center}

\begin{abstract}
\noindent
The survival and proliferation of cells and organisms
require a highly coordinated allocation of cellular resources 
to ensure the efficient synthesis of cellular components.
In particular, the total enzymatic capacity for cellular metabolism
is limited by \col{finite resources that are shared between all enzymes,}
such as cytosolic space, energy expenditure \col{for amino-acid synthesis, or micro-nutrients}.
While extensive work has been done
to study constrained optimization problems based \col{only} on stoichiometric information,
mathematical results that characterize the optimal flux in kinetic metabolic networks are still scarce. 
Here, we study constrained enzyme allocation problems with general kinetics,
using the theory of oriented matroids.
We give a rigorous proof for the fact
that \col{optimal} solutions of the non-linear optimization problem are elementary flux modes. 
This finding has significant consequences for our understanding of 
\col{optimality in metabolic networks} as well as for \col{the identification of}
metabolic switches and the computation of optimal flux \col{distributions} in kinetic metabolic networks.
\end{abstract}

\vspace{4ex}


\section{Introduction}

Living organisms are under constant evolutionary pressure to survive and reproduce in complex environments. 
As a direct consequence, cellular pathways are often assumed to be highly adapted to their respective tasks,
given the biochemical and biophysical constraints of their environment. 
Optimality principles \col{have proven to be} powerful methods
to study and understand the large-scale organization of metabolic pathways~\cite{Berkhout2012,Heinrich1996,Schuetz2012,Steuer2008,Molenaar2009}.
A variety of recent computational techniques, such as flux-balance analysis (FBA), 
seek to identify metabolic flux distributions
that maximize given objective functions, such as ATP regeneration or biomass yield,
under a set of linear constraints. 
As one of their prime merits, FBA and related stoichiometric methods,
\col{including the generalization to time-dependent metabolism}~\cite{Mahadevan2002,Antoniewicz2013},
only require knowledge of the stoichiometry of a metabolic network
-- data that are available for an increasing number of organisms
in the form of large-scale metabolic reconstructions~\cite{Oberhardt2009,Orth2010}. 

However, despite their explanatory and predictive success,
constraint-based stoichiometric methods also have inherent limits.
Specifically, FBA and related methods typically maximize stoichiometric yield.
That is, the value of a designated output flux is maximized, given a set of limiting input fluxes.
As emphasized in a number of recent studies,
the assumption of \col{maximal stoichiometric yield} is
not necessarily a universal principle of metabolic network function~\cite{Schuster2008,Schuster2011,Goel2012,Molenaar2009}.
Quite on the contrary, examples of seemingly suboptimal metabolic behavior,
at least from a stoichiometric perspective, 
are well-known for many decades.
Among the most prominent instances are the Warburg and the Crabtree effect~\cite{Warburg1924,Crabtree1928,Hsu2008}. 
Under certain circumstances,
cells utilize a fermentative metabolism rather than aerobic respiration to regenerate ATP,
despite the presence of oxygen and despite its significantly lower stoichiometric yield of ATP per amount of glucose consumed. 

To account for such seemingly suboptimal behavior,
several modifications and extensions of FBA have been developed recently.
Conventional FBA is augmented with additional principles
concerning limited cytosolic volume~\cite{Beg2007,Vazquez2008,Vazquez2010,Vazquez2011,Shlomi2011},
membrane occupancy~\cite{Zhuang2011},
and other, more general capacity constraints~\cite{Schuster2011,Goelzer2011a}.
Each of these extensions allows for new insights into suboptimal stoichiometric behavior,
and additional constraints often also induce the utilization of pathways with lower stoichiometric yield.
However, none of the modifications of FBA addresses a metabolic network as a genuine dynamical system
with particular kinetics \col{that depend} on a number of parameters.
The neglect of the dynamical nature
is a direct consequence of the extensive data requirements for parametrizing \col{enzymatic reaction rates}.
Correspondingly, and despite its importance to understand metabolic optimality,
only few mathematically rigorous results are currently available
that allow to characterize solutions of constrained non-linear optimization problems
arising from kinetic metabolic networks.

In this work, we formulate and study constrained enzyme allocation problems
in metabolic networks with general kinetics.
In particular, we are interested in enzyme distributions that maximize a designated output flux,
given a limited total enzymatic capacity.
We show that the optimal distributions of metabolic fluxes
differ from solutions obtained by FBA and related stoichiometric methods.
Most importantly, we give a rigorous proof for the fact
that optimal flux distributions are elementary flux modes. 
Therein, we make use of results from the theory of oriented matroids
that were hitherto only scarcely applied to metabolic networks~\cite{BeardBabsonCurtisQian2004},
but offer great potential to unify and advance metabolic network analysis,
as mentioned in~\cite{GagneurKlamt2004,MuellerBockmayr2013}.
Our finding has significant consequences for the understanding of metabolic \col{optimality}
as well as for the \col{efficient} computation of optimal fluxes in kinetic metabolic networks.

The paper is organized as follows.
\col{In Section~\ref{sec:def},}
we introduce kinetic metabolic networks
and state the \col{enzyme allocation problem} of interest. 
\col{In Section~\ref{sec:ex},}
we illustrate our mathematical results and the ideas \col{underlying our} proofs
by a conceptual example of a minimal metabolic network.
\col{In Section~\ref{sec:math},}
we address the connection between metabolic network analysis and the theory of oriented matroids.
\col{In particular, we reformulate the optimization problem
and show that, if the enzyme allocation problem has an optimal solution,}
then it has an optimal solution
which is an elementary flux mode.
Finally, we provide a discussion of our results
in the context of metabolic optimization problems. 

\section{Problem statement} \label{sec:def}

After introducing the necessary mathematical notation,
we define kinetic meta\-bolic networks and elementary flux modes,
and state the metabolic optimization problem that we investigate in the following.

\subsubsection*{Mathematical notation}

We denote the positive real numbers by $\RR_>$ and the non-negative real numbers by $\RR_\ge$.
For a finite index set $I$, we write $\RR^I$ for the real vector space of \col{vectors $x=(x_i)_{i \in I}$} with $x_i \in \RR$,
and $\RR^I_>$ and $\RR^I_\ge$ for the corresponding subsets.
Given $x \in \RR^I$, we write $x > 0$ if $x \in \RR^I_>$ and $x \ge 0$ if $x \in \RR^I_\ge$.
We denote the support of a vector $x \in \RR^I$ by $\supp(x) = \{ i \in I \mid x_i \neq 0 \}$.
For $x,y \in \RR^I$, we denote the component-wise (or Hadamard) product by $x \circ y \in \RR^I$,
that is, $(x \circ y)_i = x_i y_i$.

\subsubsection*{Kinetic metabolic networks}

A {\em metabolic network} $(\mathcal{S},\mathcal{R},N)$
consists of a set $\mathcal{S}$ of internal metabolites, a set $\mathcal{R}$ of reactions,
and the stoichiometric matrix $N \in \RR^{\mathcal{S} \times \mathcal{R}}$,
which contains the net stoichiometric coefficients for each metabolite $s \in \mathcal{S}$ in each reaction $\rr \in \mathcal{R}$.
The set of reactions is the disjoint union of the sets of reversible and irreversible reactions,
$\mathcal{R}_\leftrightarrow$ and $\mathcal{R}_\rightarrow$, respectively.

In the following, we assume that each reaction can be catalyzed by an enzyme.
Let $x \in \RR^\mathcal{S}_\ge$ denote the vector of metabolite concentrations,
$c \in \RR^\mathcal{R}_\ge$ the vector of enzyme concentrations,
and $p \in \RR^\mathcal{P}$ a vector of parameters
such as turnover numbers, equilibrium constants, and Michaelis-Menten constants.
We write the vector of rate functions
$v \colon \RR_\ge^\mathcal{S} \times \RR_\ge^\mathcal{R} \times \RR^\mathcal{P} \to \RR^\mathcal{R}$ as
\[
v(x;c,p) = c \circ \ka(x,p)
\]
with a function
$\ka \colon \RR^\mathcal{S}_\ge \times \RR^\mathcal{P} \to \RR^\mathcal{R}$.
In other words, 
each reaction rate $v_\rr$ is the product of the corresponding enzyme concentration $c_\rr$
with a particular kinetics $\ka_\rr$.

A {\em kinetic metabolic network} $(\mathcal{S},\mathcal{R},N,v)$ is a metabolic network $(\mathcal{S},\mathcal{R},N)$
together with rate functions $v$ as defined above.
The dynamics of $(\mathcal{S},\mathcal{R},N,v)$ is determined by the ODEs
\[
\dd{x}{t} = N v(x;c,p) .
\]
A steady state $\bar{x} \in \RR^\mathcal{S} $
and the corresponding steady-state flux $\bar{v}=v(\bar{x};c,p) \in \RR^\mathcal{R}$
are determined by
\[
0 = N \bar{v} .
\]

\subsubsection*{Elementary flux modes}

A {\em flux mode} is a non-zero steady-state flux $f \in \RR^\mathcal{R}$
with non-negative components for all irreversible reactions.
In other words,
a flux mode is a non-zero element of the {\em flux cone}
\[
\FC = \{ f \in \RR^\mathcal{R} \mid N f = 0 \text{ and } f_\rr \ge 0 \text{ for all } \rr \in \mathcal{R}_\rightarrow \} .
\]

An {\em elementary flux mode} (EFM) $e \in \RR^\mathcal{R}$
is a flux mode with minimal support:
\begin{equation} \label{efm}
f \in \FC \text{ with } f \neq 0 \text{ and } \supp(f) \subseteq \supp(e) \Rightarrow \supp(f) = \supp(e). \tag{efm}
\end{equation}
In fact, $\supp(f) = \supp(e)$ further implies $f = \lambda \, e$ with $\lambda>0$.
Otherwise, one can construct another flux mode from $e$ and $f$ with smaller support.
As a consequence, there can be only finitely many EFMs (up to multiplication with a positive scalar).
For references on elementary flux modes and related computational issues, see
\cite{SchusterHilgetag1994,SchusterHilgetagWoodsFell2002,KlamtStelling2003,GagneurKlamt2004,WagnerUrbanczik2005,LarhlimiBockmayr2009}.

\subsubsection*{Enzyme allocation problem}

We are now in a position to state the optimization problem
that we study in this work.

{\em
\col{Let} $(\mathcal{S},\mathcal{R},N,v)$ be a kinetic metabolic network \col{with flux cone $\FC$.}
\col{Fix} a reaction $\ro \in \mathcal{R}$,
positive weights $w \in \RR^\mathcal{R}_>$,
\col{a subset $X \subseteq \RR^\mathcal{S}_\ge$,
and parameters $p \in \RR^\mathcal{P}$.}
Maximize the component $\bar{v}_\ro > 0$ of the steady-state flux $\bar{v} = c \mal \ka(\bar{x},p)$
by varying \col{the steady-state metabolite concentrations $\bar{x} \in X$}
and the enzyme concentrations $c \in \RR^\mathcal{R}_\ge$.
\col{Thereby, fix} the weighted sum of enzyme concentrations
\col{and require the steady-state flux to be a flux mode}:
\begin{subequations} \label{original}
\begin{equation}
\max_{\bar{x} \in X, \, c \in \RR^\mathcal{R}_\ge} \bar{v}_\ro
\end{equation}
subject to
\begin{gather}
\sum_{\rr \in \mathcal{R}} w_\rr \, c_\rr = c^\tot , \\
\col{\bar{v} \in \FC, \, \bar{v}_\ro > 0 . }
\end{gather} 
\end{subequations}
}

\col{Note that the constraint $\bar{v} \in \FC$ implies}
$0 = N \bar{v} = N \bar{v}(c \circ \ka(\bar{x},p))$
which is non-linear, in general.

\col{
Further, note that the enzyme allocation problem may be unfeasible,
in particular, the constraint~(\ref{original}c) may be unsatisfiable.
In this case, the flux cone and the kinetics are ``incompatible''.
Moreover, even if the problem is feasible, the maximum may not be attained at finite metabolite concentrations.
On the other hand,
if the problem is feasible, the kinetics is continuous in $\bar{x}$,
and $X$ is compact (bounded and closed),
then the maximum is attained.
}%

Problem~\eqref{original} defines a very general metabolic optimization problem,
in which the set of enzyme concentrations is adjusted 
in order to maximize a specific metabolic flux within the steady-state flux vector. 
In this respect, the weighted sum (\ref{original}b) may encode different enzymatic constraints,
such as limited cellular or membrane surface space,
limited nitrogen or transition metal availability,
as well as other constraints for the abundance of certain enzymes. 
In each case, the weight factors denote the fraction of the resource used per unit enzyme. 
Likewise, the flux component $\bar{v}_\ro$ may stand for diverse metabolic processes,
ranging from the synthesis rate of a particular product within a specific pathway
to the rate of overall cellular growth.
The optimization problem seeks to identify the maximum value of $\bar{v}_\ro$,
the associated enzyme and steady-state metabolite concentrations, $c$ and $\bar{x}$,
as well as the corresponding flux $\bar{v}$.

\section{A conceptual example} \label{sec:ex}

\begin{figure}
\begin{center}
\includegraphics[width=\textwidth]{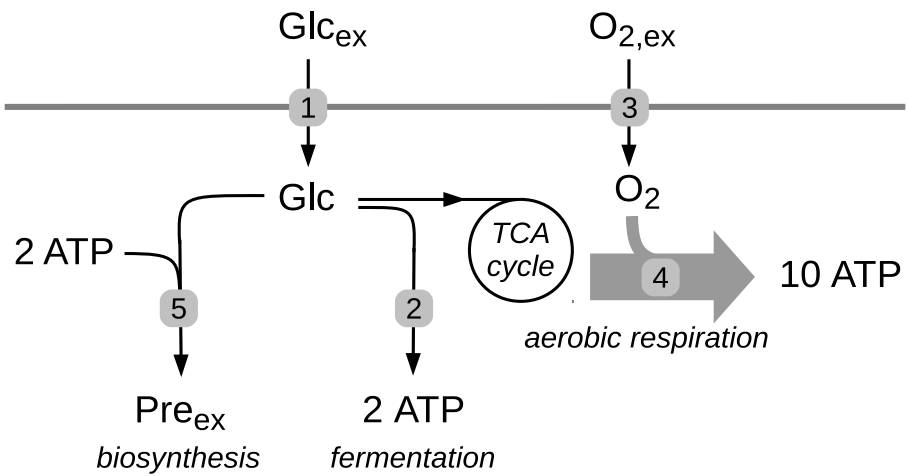}
\end{center}
\caption{
A minimal metabolic network.
The pathway consists of two uptake reactions (1,3), two intracellular conversions (2,4),
and the formation of a precursor molecule (5).
The resulting EFMs are $e^{1} = \left( 2,1,0,0,1 \right)$
and $e^{2} = \left( \frac{6}{5},0,\frac{1}{5},\frac{1}{5},1 \right)$
with overall reactions \mbox{2 \ce{Glc_\xx} $\to$ 1 \ce{Pre_\xx}} (fermentation)
and \mbox{6 \ce{Glc_\xx} + 1 \ce{O_{2,\xx}} $\to$ 5 \ce{Pre_\xx}} (respiration).
\label{fig:network}}
\end{figure}

In order to illustrate our mathematical results and the ideas in its proofs,
we first study a minimal metabolic network for the production of a precursor molecule from glucose.
In particular, we consider two alternative pathways:
fermentation (low yield) and respiration (high yield), cf.\ Figure~\ref{fig:network}.
The actual metabolic network $(\mathcal{S},\mathcal{R},N)$ is further simplified
and involves the internal metabolites $\mathcal{S} = \{ \ce{Glc}, \ce{O_2}, \ce{ATP} \}$,
the set of reactions $\mathcal{R}$ consisting of
\begin{alignat*}{3}
& 1  \, &:\quad& \ce{Glc_\xx} && \rightleftharpoons \ce{Glc} \\
& 2  &:\quad& \ce{Glc} && \rightleftharpoons \ce{2 ATP} \\
& 3  &:\quad& \ce{O_{2,\xx}} && \rightleftharpoons \ce{O_2} \\
& 4  &:\quad& \ce{Glc} + \ce{O_2} && \rightleftharpoons \ce{10 ATP} \\
& 5  &:\quad& \ce{Glc} + \ce{2 ATP} && \rightarrow \ce{Pre_\xx} ,
\end{alignat*}
and the resulting stoichiometric matrix
\[
N=
\bordermatrix{
 & 1 & 2 & 3 & 4 & 5 \cr
\ce{Glc} & +1 & -1 &  0 & -1 & -1 \cr
\ce{O_2} &  0 &  0 & +1 & -1 &  0 \cr
\ce{ATP} &  0 & +2 &  0 &+10 & -2 \cr
} .
\]

The external substrates/products \ce{Glc_\xx}, \ce{O_{2,\xx}}, \ce{Pre_\xx}
do not appear in $(\mathcal{S},\mathcal{R},N)$,
but their constant concentrations can enter the rate functions as parameters.
In a kinetic metabolic network $(\mathcal{S},\mathcal{R},N,v)$,
the rate of reaction $\rr$ is given as 
\[ 
v_\rr = c_\rr \, \ka_\rr , 
\]
that is, as a product of the corresponding enzyme concentration $c_\rr$ and a particular kinetics $\ka_\rr$.
In this example, we use the kinetics
\begin{align*}
\ka_1 &= k_1 \left( [\ce{Glc_\xx}] - \col{K_1} [\ce{Glc}] \right) \\
\ka_2 &= k_2 \left( [\ce{Glc}] - \col{K_2} [\ce{ATP}] \right) \\
\ka_3 &= k_3 \left( [\ce{O_{2,\xx}}] - \col{K_3} [\ce{O_2}] \right) \\
\ka_4 &= k_4 \left( [\ce{Glc}] [\ce{O_2}] - \col{K_4} [\ce{ATP}] \right) \\
\ka_5 &= k_5 \, [\ce{Glc}] [\ce{ATP}] ,
\end{align*}
however, our \col{mathematical results} do not depend on the kinetics.
In vector notation, we write
\[
v(x;c,p) = c \circ \ka(x,p) ,
\]
thereby introducing the concentrations $x$ of the internal metabolites and the parameters $p$:
\begin{align*}
x &= ([\ce{Glc}],[\ce{O_2}],[\ce{ATP}])^T , \\
p &= ([\ce{Glc_\xx}],[\ce{O_{2,\xx}}];k_1,k_2,k_3,k_4,k_5;K_1,K_2,K_3,K_4)^T .
\end{align*}
The dynamics of the network is governed by the ODEs
\[
\dd{x}{t} = N v(x;c,p) .
\]
A steady state $\bar{x}$ and the corresponding steady-state flux $\bar{v} = v(\bar{x};c,p)$ are determined by
\[
0 = N \bar{v} .
\]

The goal is to maximize the production rate of the precursor \ce{Pre_\xx},
that is, the component $\bar{v}_5$ of the steady-state flux \col{$\bar{v} = c \mal \ka(\bar{x},p)$},
by varying \col{the steady-state metabolite concentrations $\bar{x} \in \RR^3_\ge$}
and the enzyme concentrations $c \in \RR^5_\ge$.
\col{Thereby}, the (weighted) sum of enzyme concentrations \col{is fixed}
\col{and the steady-state flux must be a flux mode}:
\begin{subequations} \label{example}
\begin{equation}
\max_{\bar{x},c} \, \bar{v}_5
\end{equation}
subject to
\begin{gather}
\sum_{\rr=1}^5 c_\rr = c^\tot , \\
\col{N \bar{v} = 0 , \, \bar{v}_5 > 0 . }
\end{gather}
\end{subequations}
For convenience, we use equal weights in the sum constraint.

To simplify the problem, we consider a restriction on the steady-state metabolite concentrations $\bar{x}$,
in particular, we require $\ka(\bar{x},p) > 0$.
In chemical terms, we assume the thermodynamic feasibility
of a situation where all reactions can proceed from left to right.
Since $\bar{v} = c \circ \ka$, this implies $\bar{v} \ge 0$.
That is, every component of the steady-state flux is non-negative;
\col{in fact}, it is zero if and only if the corresponding enzyme concentration is zero.

Every \col{feasible} steady-state flux $\bar{v}$ is a flux mode.
In particular,
\col{$\bar{v} = \bar{v}_5 \, f$ where $f$ is a flux mode with $f \ge 0$ and $f_5 = 1$.}
From $\bar{v} = c \circ \ka = \bar{v}_5 \, f$,
we further obtain
\[
c_\rr = \bar{v}_5 \, \frac{f_\rr}{\ka_\rr} .
\]
Now, we can rewrite the constraint on the enzyme concentrations as
\[
c^\tot = \sum_{r=1}^5 c_\rr = \bar{v}_5 \sum_{r=1}^5 \frac{f_\rr}{\ka_\rr} .
\]
Instead of maximizing $\bar{v}_5$,
we can minimize $\frac{c^\tot}{\bar{v}_5} = \sum_{r=1}^5 \frac{f_\rr}{\ka_\rr(\bar{x},p)}$
by varying the steady-state metabolite concentrations $\bar{x}$
and the flux mode $f$.
\col{Hence}, the enzyme allocation problem~\col{\eqref{example} with the restriction $\ka(\bar{x},p) > 0$}
is equivalent to:
\begin{subequations} \label{examplerestricted}
\begin{equation}
\min_{\bar{x},f} \, \sum_{r=1}^5 \frac{f_\rr}{\ka_\rr(\bar{x},p)}
\end{equation}
subject to
\begin{gather}
\ka(\bar{x},p) > 0 , \, \col{f \ge 0 ,} \\
\col{N f = 0}, \, f_5 = 1 .
\end{gather}
\end{subequations}

\col{As shown in Subsection~\ref{subsec:sign},}
every flux mode $f \ge 0$ is a non-negative linear combination of elementary flux modes (EFMs) $e \ge 0$.
In fact, there are two such EFMs,
\begin{align*}
e^{1} &= \left( 2,1,0,0,1 \right)^T \\
e^{2} &= \left( {\textstyle \frac{6}{5},0,\frac{1}{5},\frac{1}{5},1 } \right)^T ,
\end{align*}
representing fermentation and respiration,
and hence
\[
f = \alpha_1 \, e^1 + \alpha_2 \, e^2 \quad \text{with } \alpha_1, \alpha_2 \ge 0 .
\]
Note that we have scaled the EFMs $e$ such that $e_5=1$.
The condition $f_5 = 1$ implies $\alpha_1 + \alpha_2 = 1$.
As a result, we obtain another equivalent formulation of the restricted enzyme allocation problem:
\begin{subequations}
\begin{equation}
\min_{\bar{x},\alpha_1,\alpha_2} \, \sum_{r=1}^5 \frac{\alpha_1 \, e^1_\rr + \alpha_2 \, e^2_\rr}{\ka_\rr(\bar{x},p)}
\end{equation}
subject to
\begin{gather}
\ka(\bar{x},p) > 0 , \\
\alpha_1 + \alpha_2 = 1 .
\end{gather}
\end{subequations}

We observe that the objective function is linear in $\alpha_1$ and $\alpha_2$:
\[
\alpha_1 \underbrace{\sum_{r=1}^5 \frac{e^1_\rr}{\ka_\rr(\bar{x},p)}}_{g_1}
+ \, \alpha_2 \underbrace{\sum_{r=1}^5 \frac{e^2_\rr}{\ka_\rr(\bar{x},p)}}_{g_2}
\]
with
\begin{align*}
g_1(\bar{x},p) &= \frac{2}{\ka_1(\bar{x},p)} + \frac{1}{\ka_2(\bar{x},p)} + \frac{1}{\ka_5(\bar{x},p)} \\
g_2(\bar{x},p) &= \frac{\textstyle \frac{6}{5}}{\ka_1(\bar{x},p)} + \frac{\textstyle \frac{1}{5}}{\ka_3(\bar{x},p)}
+ \frac{\textstyle \frac{1}{5}}{\ka_4(\bar{x},p)} + \frac{1}{\ka_5(\bar{x},p)}.
\end{align*}
Clearly, $g_1(\bar{x},p) > 0$ and $g_2(\bar{x},p) > 0$, since $\ka(\bar{x},p) > 0$.
\col{Assume that the minima of $g_1$ and $g_2$ are attained at $\bar{x}^1$ and $\bar{x}^2$, respectively.
That is,}
$\hat{g}_1 = \min_{\bar{x}} g_1(\bar{x},p)$ $\col{= g_1(\bar{x}^1,p)}$
and $\hat{g}_2 = \min_{\bar{x}} g_2(\bar{x},p)$ $\col{= g_2(\bar{x}^2,p)}$.
If $\hat{g}_1 < \hat{g}_2$, then the objective function
\col{attains its minimum $\hat{g}_1$ at $\bar{x}=\bar{x}^1$,} $\alpha_1 = 1$ and $\alpha_2 = 0$,
that is, for $f = e^1$.
\col{
To see this, assume $\alpha_2>0$;
then, for all $\bar{x}$,
\[
\alpha_1 g_1(\bar{x},p) + \alpha_2 g_2(\bar{x},p)
\ge \alpha_1 \hat{g}_1 + \alpha_2 \hat{g}_2
> \alpha_1 \hat{g}_1 + \alpha_2 \hat{g}_1
= (\alpha_1 + \alpha_2) \, \hat{g}_1
= \hat{g}_1 .
\]
}%
Conversely, if $\hat{g}_1 > \hat{g}_2$, the minimum is attained for $f = e^2$,
and finally, if $\hat{g}_1 = \hat{g}_2$, both $f=e^1$ and $f=e^2$ are optimal.
In the degenerate case where $\hat{g}_1 = g_1(\bar{x}^0,p) = g_2(\bar{x}^0,p) = \hat{g}_2$ at the same minimum point~$\bar{x}^0$,
any $f = \alpha_1 \, e^1 + \alpha_2 \, e^2$ (with $\alpha_1,\alpha_2\ge0 $ and $\alpha_1 + \alpha_2 = 1$) is optimal.

We can summarize our result as follows:
generically,
the steady-state flux $\bar{v}$ related to an optimal solution of the restricted enzyme allocation problem~\eqref{examplerestricted}
is an EFM.
The same holds for all appropriate restrictions and hence for the full enzyme allocation problem~\eqref{example}.

\begin{figure}
\begin{center}
\includegraphics[width=0.75\textwidth]{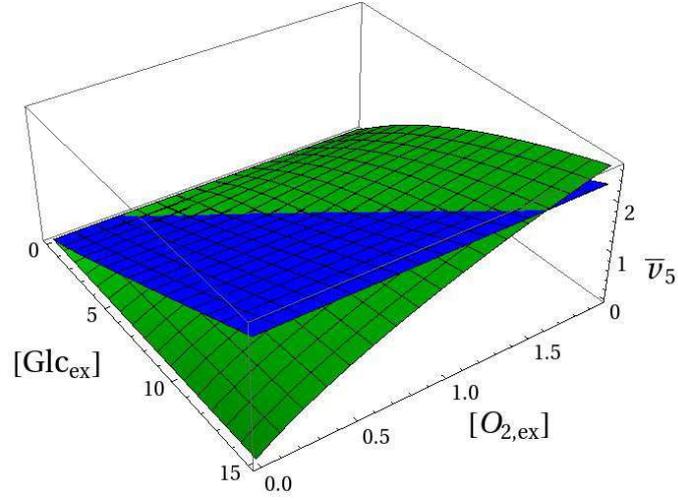}
\end{center}
\caption{
Maximum production rate $\bar{v}_5$ 
as a function of external substrate concentrations $[\ce{Glc_\xx}]$ and $[\ce{O_{2,\xx}}]$
plotted for EFM $e^1$ (fermentation, blue) and EFM $e^2$ (respiration, green).
Parameters: $k_1=k_2=k_3=k_4=k_5=1$, $K_1=K_3=1$, $K_2=K_4=0.1$.
\label{fig:1}
}
\end{figure}

\begin{figure}
\begin{center}
\includegraphics[width=0.3\textwidth]{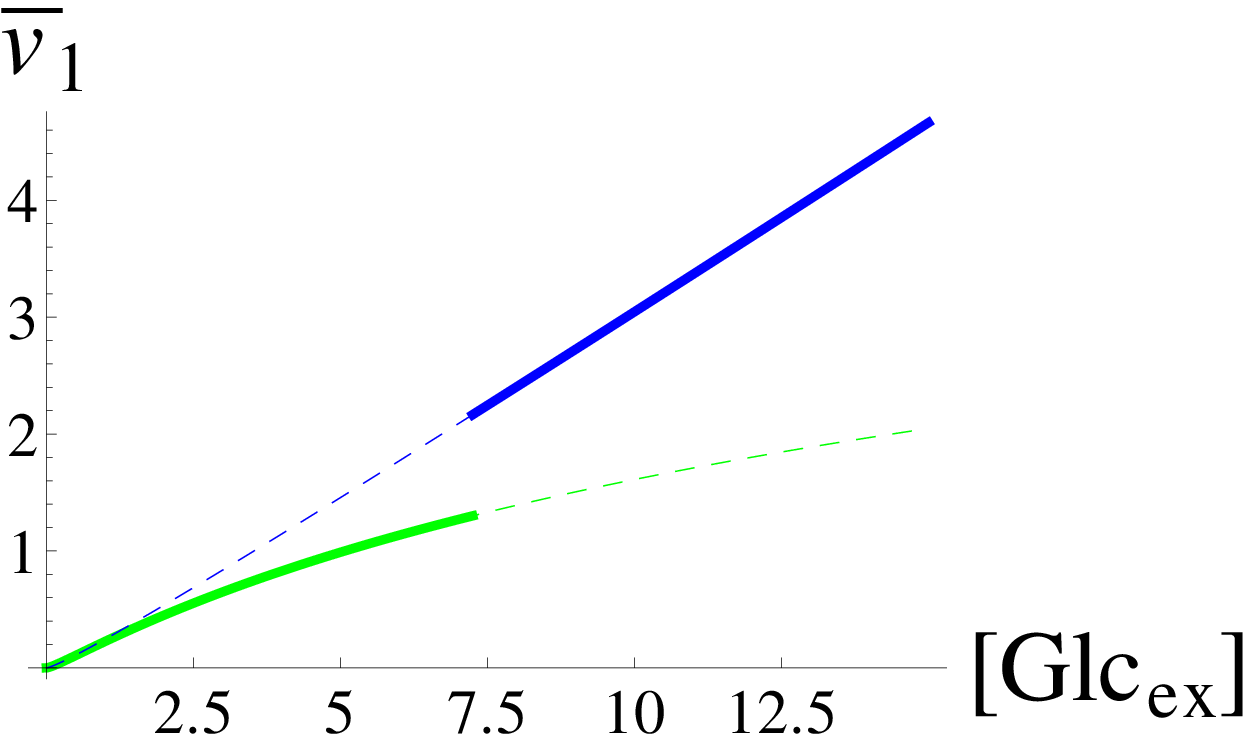} \hspace{0.25cm}
\includegraphics[width=0.3\textwidth]{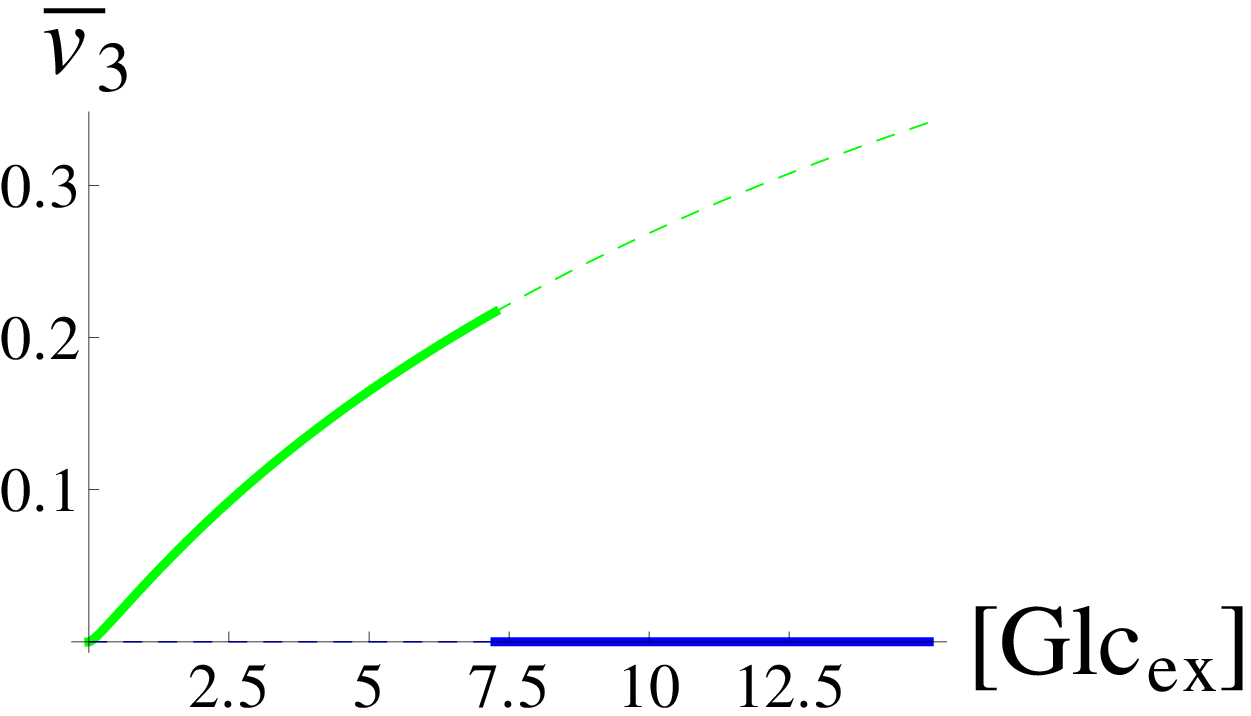} \hspace{0.25cm}
\includegraphics[width=0.3\textwidth]{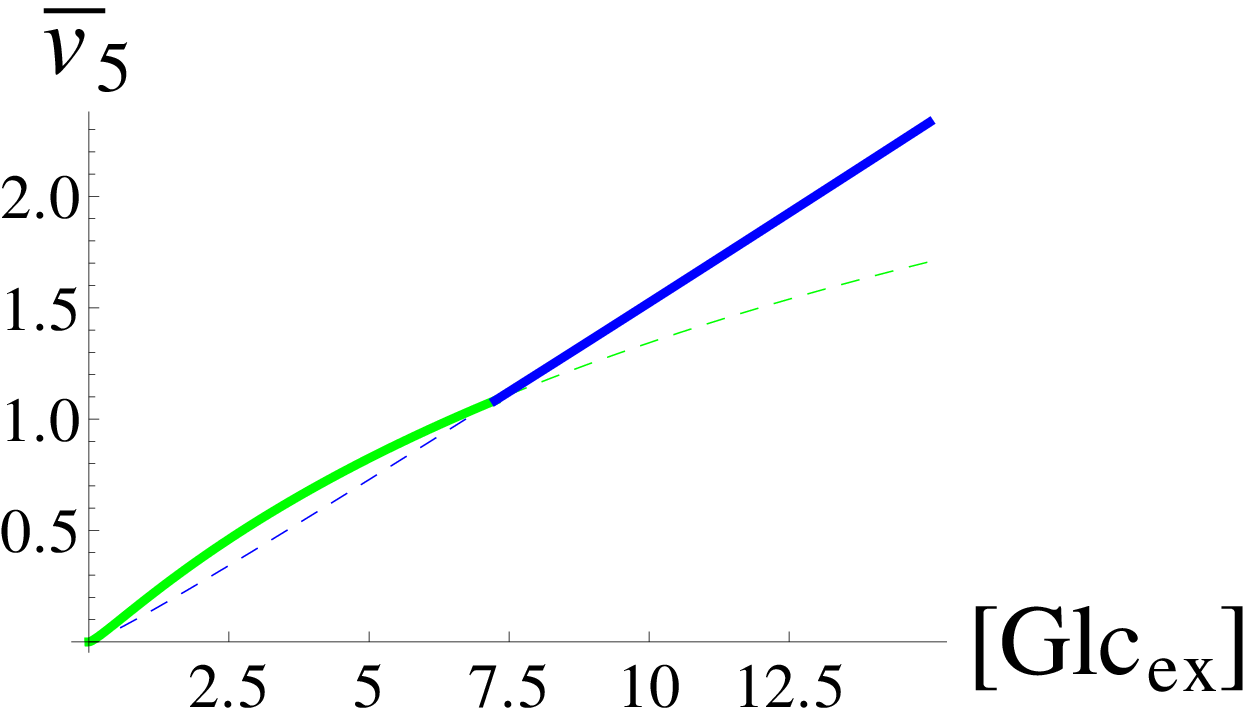} \vspace{0.5cm}

\includegraphics[width=0.3\textwidth]{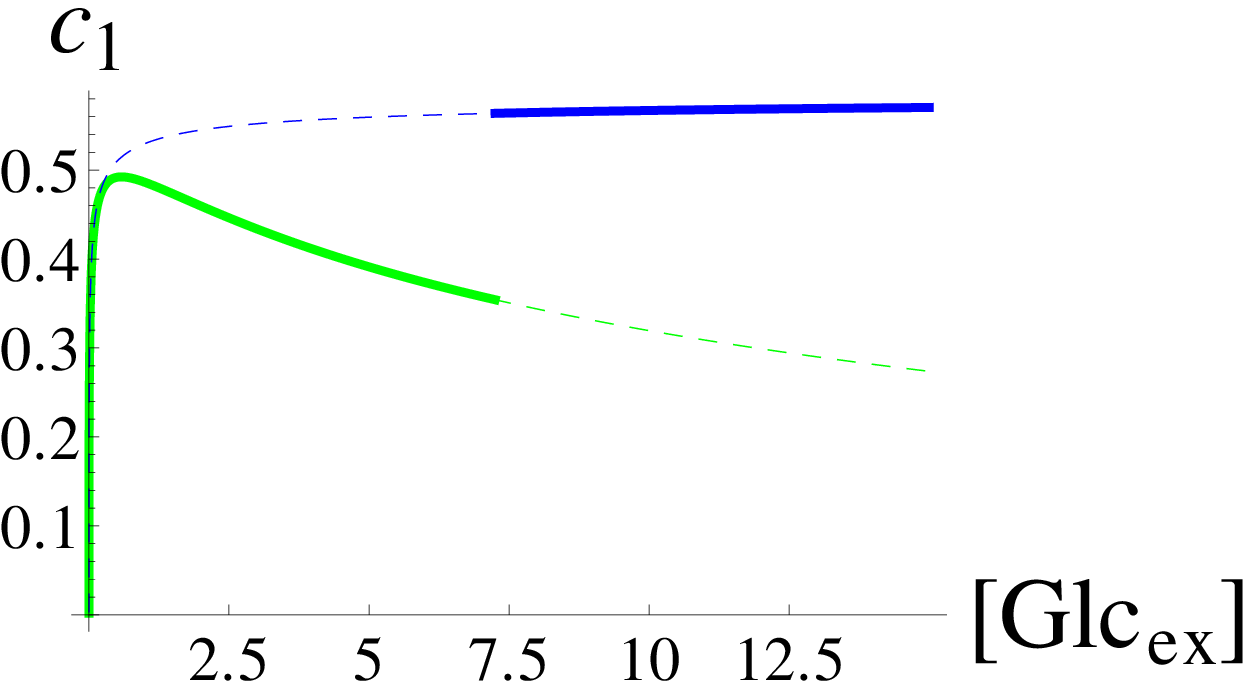} \hspace{0.25cm}
\includegraphics[width=0.3\textwidth]{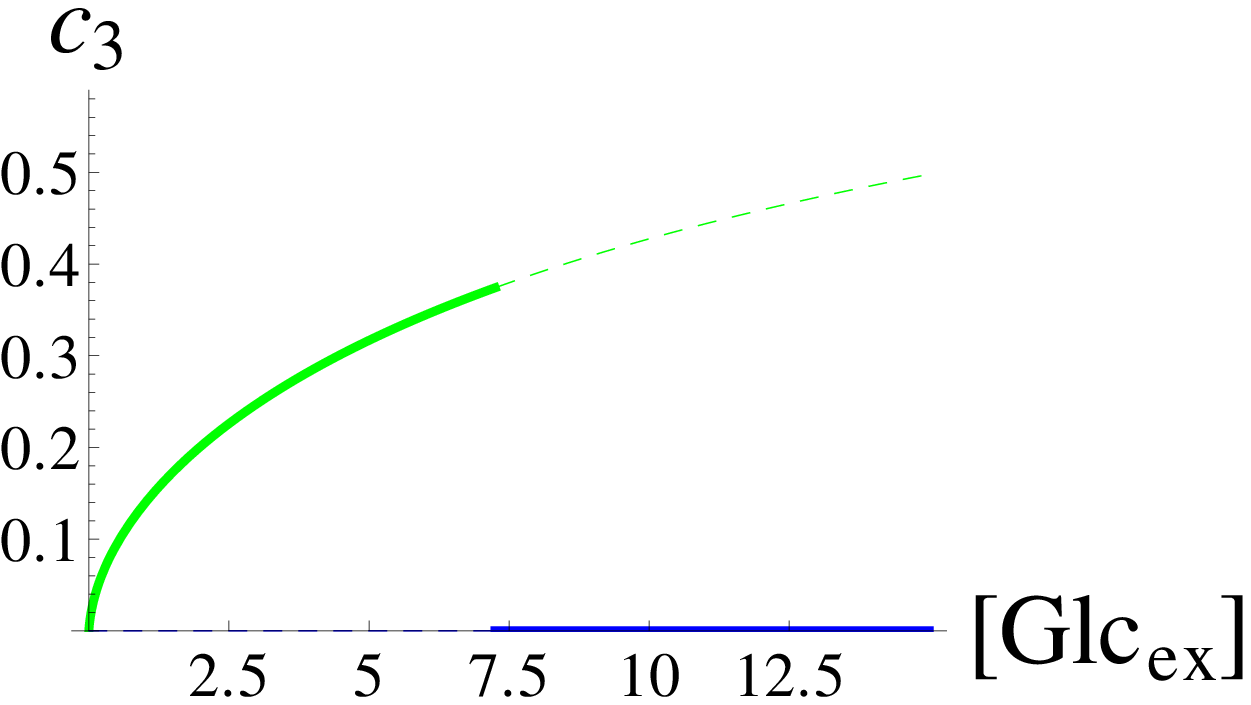} \hspace{0.25cm}
\includegraphics[width=0.3\textwidth]{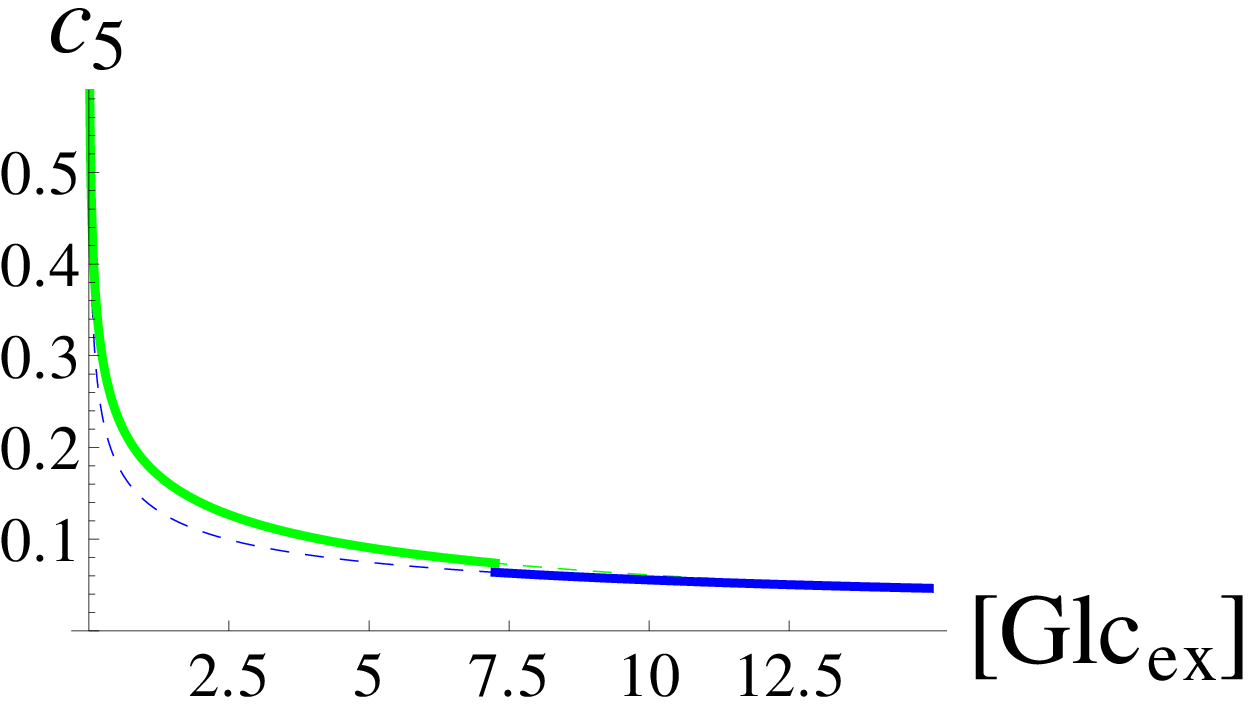} \vspace{0.5cm}

\includegraphics[width=0.3\textwidth]{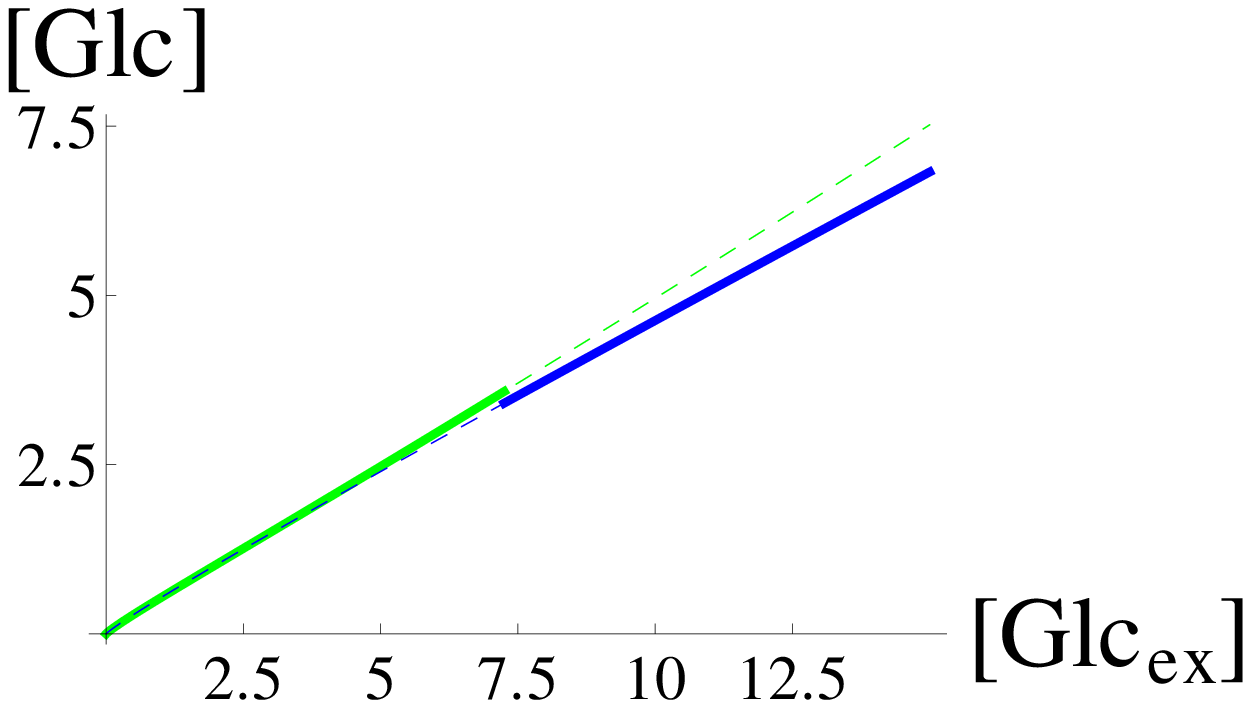}\hspace{0.25cm}
\includegraphics[width=0.3\textwidth]{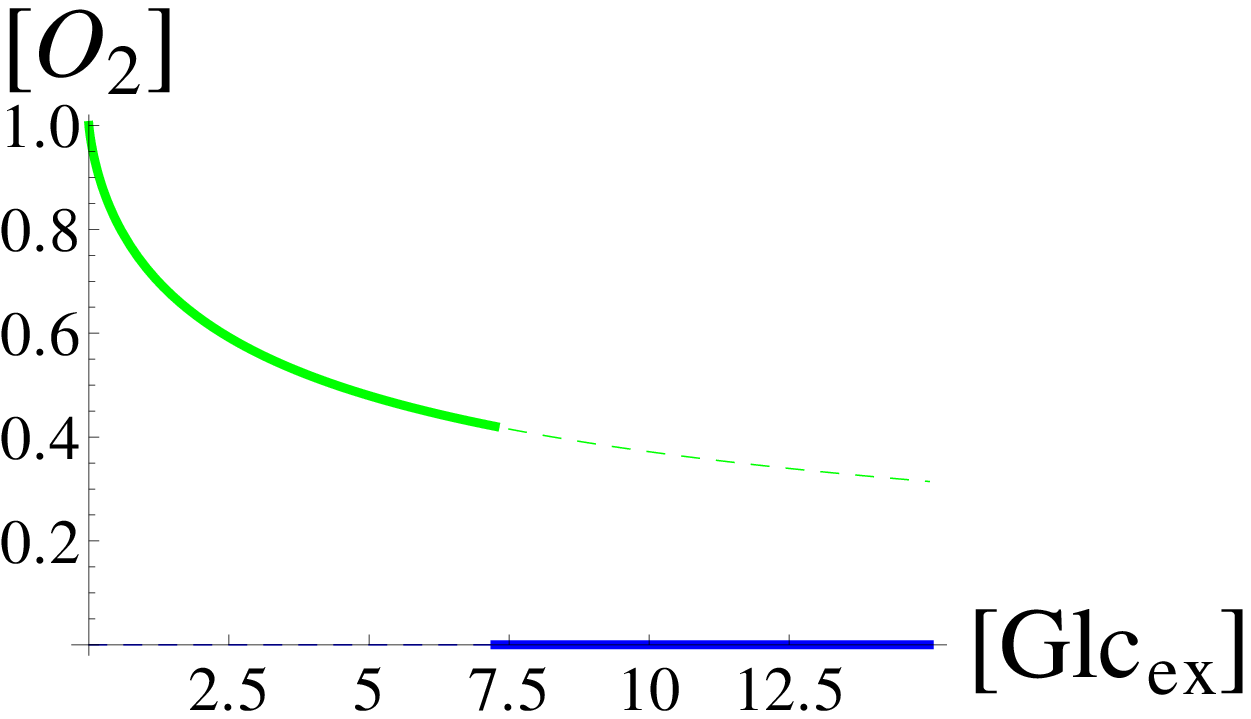} \hspace{0.25cm}
\includegraphics[width=0.3\textwidth]{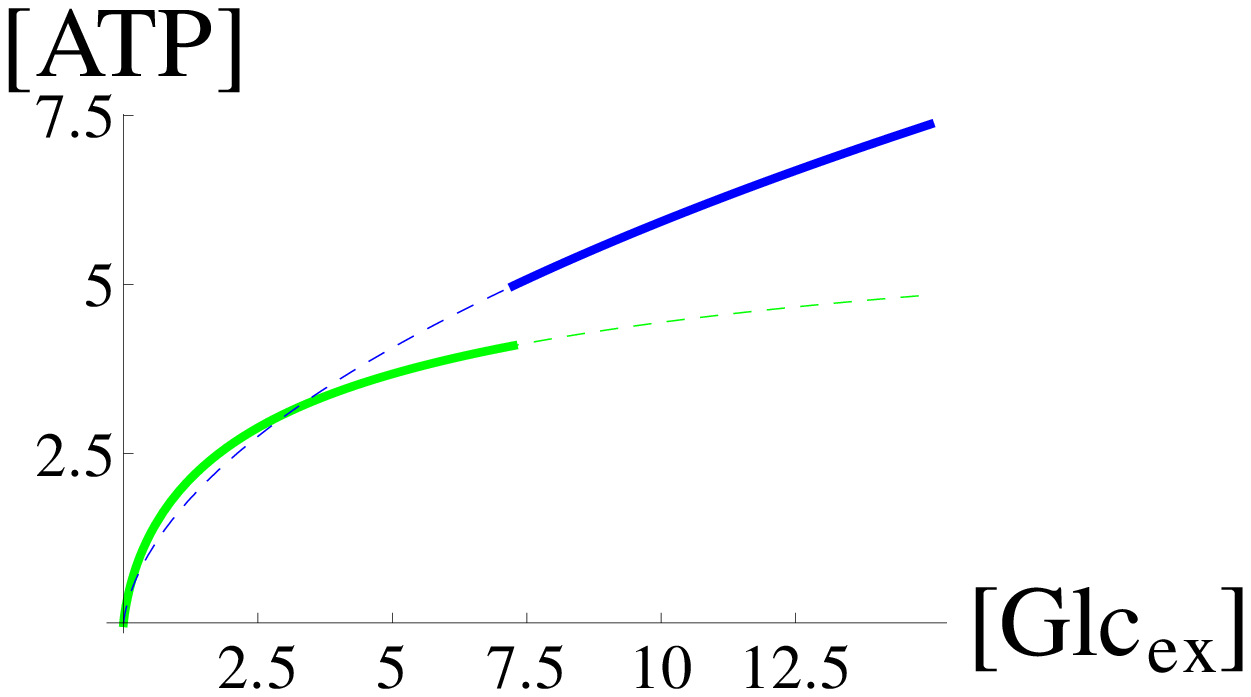}
\end{center}
\caption{
Optimal fluxes $\bar{v}_1,\bar{v}_3,\bar{v}_5$, enzyme concentrations $c_1,c_3,c_5$,
and internal metabolite concentrations $[\ce{Glc}],[\ce{O2}],[\ce{ATP}]$
as functions of external substrate concentration $[\ce{Glc_\xx}]$
plotted for EFM $e^1$ (fermentation, blue) and EFM $e^2$ (respiration, green).
With increasing $[\ce{Glc_\xx}]$,
the optimal solution switches from respiration (thick green lines) to fermentation (thick blue lines).
Parameters: $[\ce{O_{2,\xx}}]=1$, $k_1=k_2=k_3=k_4=k_5=1$, $K_1=K_3=1$, $K_2=K_4=0.1$.
\label{fig:2}
}
\end{figure}

For variable external substrate concentrations $[\ce{Glc_\xx}]$ and $[\ce{O_{2,\xx}}]$,
we are interested in which EFM is optimal and when a switch between EFMs occurs.
To this end, we determine the optimal solution for each EFM.
The optimization problem restricted to EFM $e^1$ is equivalent to
\[
\min_{\bar{x}} \, g_1(\bar{x},p)
\]
subject to
\[
\ka_1(\bar{x},p)>0, \, \ka_2(\bar{x},p)>0, \, \ka_5(\bar{x},p)>0 .
\]
\col{In EFM $e^1$, reactions 3 and 4 do not carry any flux,} that is, $e^1_3=e^1_4=0$.
\col{Hence, the corresponding} enzyme concentrations are zero,
that is, $c_3=c_4=0$,
and there are no constraints involving $\ka_3(\bar{x},p)$ and $\ka_4(\bar{x},p)$.
From the optimal metabolite concentrations $\bar{x}$,
we determine the optimal enzyme concentrations $c_1$, $c_2$, and $c_5$ as
\[
c_\rr = \bar{v}_5 \, \frac{e^1_\rr}{\ka_\rr}
= c^\tot \, \frac{\frac{e^1_\rr}{\ka_\rr(\bar{x},p)}}{\sum_{r=1,2,5} \frac{e^1_\rr}{\ka_\rr(\bar{x},p)}} .
\]

Explicitly, the optimization problem for EFM $e^1$ amounts to
\[
\min_{[\ce{Glc}],[\ce{ATP}]} \left(
\frac{2}{k_1 \left( [\ce{Glc_\xx}] - K_1 [\ce{Glc}] \right)}
+ \frac{1}{k_2 \left( [\ce{Glc}] - K_2 [\ce{ATP}] \right)}
+ \frac{1}{k_5 \, [\ce{Glc}] [\ce{ATP}]}
\right)
\]
subject to
\[
[\ce{Glc_\xx}] - K_1 [\ce{Glc}] > 0 , \,
[\ce{Glc}] -  K_2 [\ce{ATP}] > 0 , \,
[\ce{Glc}] [\ce{ATP}] > 0 ,
\]
where we omit the bar over the steady-state metabolite concentrations.
\col{From the optimal metabolite concentrations $[\ce{Glc}],[\ce{ATP}]$,}
we determine the optimal enzyme concentrations $c_1$, $c_2$, and $c_5$.
\col{For example,}
\[
c_1 = c^\tot \,
\frac{ \frac{2}{k_1 \left( [\ce{Glc_\xx}] - K_1 [\ce{Glc}] \right)} }
{ \frac{2}{k_1 \left( [\ce{Glc_\xx}] - K_1 [\ce{Glc}] \right)}
+ \frac{1}{k_2 \left( [\ce{Glc}] - K_2 [\ce{ATP}] \right)}
+ \frac{1}{k_5 \, [\ce{Glc}] [\ce{ATP}]} } .
\]
\col{The optimization problem restricted to EFM $e^2$ is treated analogously.}

Finally,
we vary $[\ce{Glc_\xx}]$ and $[\ce{O_{2,\xx}}]$,
solve the restricted optimization problems for EFMs $e^1$ and $e^2$,
and compare the resulting maximum values of $\bar{v}_5$,
cf.\ Figure~\ref{fig:1}.
Clearly, the optimal solution of the enzyme allocation problem switches between EFMs $e^1$ and $e^2$
which involves a discontinuous change of enzyme and metabolite concentrations,
cf.\ Figure~\ref{fig:2}, where we fix $[\ce{O_{2,\xx}}]=1$ and vary $[\ce{Glc_\xx}]$.

\section{Mathematical results} \label{sec:math}

We reformulate the enzyme allocation problem~\eqref{original}
and characterize its solutions.
To this end, we employ concepts from the theory of oriented matroids
like elementary vectors, sign vectors, and conformal sums.

Realizable oriented matroids arise from vector subspaces.
Essentially, a realizable oriented matroid is the set of sign vectors of a subspace
or, equivalently, all sign vectors with minimal support.
Abstract oriented matroids can be characterized by axiom systems for \mbox{(co-)vectors} (satisfied by the sign vectors of a subspace),
\mbox{(co-)circuits} (satisfied by the sign vectors with minimal support), or, equivalently, chirotopes.
For an introduction to oriented matroids, we refer to the survey~\cite{Richter-GebertZiegler1997},
the textbooks~\cite{BachemKern1992} and~\cite[Chapters~6 and 7]{Ziegler1995},
and the encyclopedic treatment~\cite{BjornerLasSturmfelsWhiteZiegler1999}.

In applications to metabolic network analysis,
the involved oriented matroids are realizable.
For example, the sign vector of a thermodynamically feasible steady-state flux
must be orthogonal to all (internal) circuits~\cite{BeardBabsonCurtisQian2004}.
In the original proof, the circuit axioms for oriented matroids are used explicitly;
however, the result also follows from basic facts about the orthogonality of sign vectors of subspaces~\cite[Chapter~6]{Ziegler1995};
alternatively, it can be proved using linear programming duality~\cite{Mueller2012,Noor2012}.
We note that
oriented matroids also appear in the study of directed hypergraph and Petri net models of biochemical reactions~\cite{Oliveiraetal2001}
and in the theory of chemical reaction networks with generalized mass action kinetics~\cite{MuellerRegensburger2012}.

\subsection{Elementary vectors}

An {\em elementary vector} (EV) $e \in \RR^\mathcal{R}$ of a \col{vector} subspace $S \subseteq \RR^\mathcal{R}$
is a non-zero vector with minimal support~\cite{Rockafellar1969}:
\begin{equation} \label{ev}
f \in S \text{ with } f \neq 0 \text{ and } \supp(f) \subseteq \supp(e) \Rightarrow \supp(f) = \supp(e) . \tag{ev}
\end{equation}

It is easy to see that EFMs are exactly those EVs of $\ker(N)$
that are flux modes.
To our knowledge, this fact has not been clarified before.
\begin{lem} \label{lem:efmev}
Let $(\mathcal{S},\mathcal{R},N)$ be a metabolic network and $e \in \RR^\mathcal{R}$.
The following statements are equivalent:
\begin{itemize}
\item[(i)]
$e$ is an EFM.
\item[(ii)]
$e$ is an EV of $\,\ker(N)$ and a flux mode.
\end{itemize}
\end{lem}
\begin{proof} (i) $\Rightarrow$ (ii):
We have to show that EFM~$e$ is an EV of $\ker(N)$.
Suppose \eqref{ev} is violated, that is, there exists $f \in \ker(N)$ with $f \neq 0$ and $\supp(f) \subset \supp(e)$.
If $f_\rr \ge 0$ for all $\rr \in \mathcal{R}_\rightarrow$, then $f \in \FC$ in contradiction to \eqref{efm}.
Otherwise, consider $f' = e + \lambda \, f$ with the largest scalar $\lambda>0$
such that $f'_\rr \ge 0$ for all $\rr \in \mathcal{R}_\rightarrow$.
Then, $f' \in \FC$ with $f' \neq 0$ and $\supp(f') \subset \supp(e)$ in contradiction to \eqref{efm}.
(ii) $\Rightarrow$ (i): Let $e$ be a flux mode and an EV of $\ker(N)$.
Clearly, \eqref{ev} implies \eqref{efm}, since $f \in \FC$ implies $f \in \ker(N)$.
\end{proof}

\subsection{Sign vectors and conformal sums} \label{subsec:sign}

We define the {\em sign vector} $\sigma(x) \in \{-,0,+\}^I$ of a vector $x \in \RR^I$
by applying the sign function component-wise. 
The relations $0<-$ and $0<+$ induce a partial order on $\{-,0,+\}^I$:
we write $X \le Y$ for $X,Y \in \{-,0,+\}^I$,
if the inequality holds component-wise.
For $x,y \in \RR^I$,
we say that $x$ {\em conforms} to $y$, if $\sigma(x) \le \sigma(y)$.
Analogously, for $x \in \RR^I$ and $X \in \{-,0,+\}^I$,
we say that $x$ conforms to $X$, if $\sigma(x) \le X$.

The following fundamental result about vectors and EVs will be rephrased for flux modes and EFMs.
For a proof, see~\cite[Theorem~1]{Rockafellar1969},~\cite[Proposition~5.35]{BachemKern1992} or~\cite[Lemma~6.7]{Ziegler1995}.
\begin{thm} \label{thm:conform}
Let $S \subseteq \RR^{\mathcal{R}}$ be a subspace.
Then every vector $f \in S$ is a conformal sum of EVs.
That is,
there exists a finite set $E$ of EVs conforming to $f$ such that
\[
f = \sum_{e \in E} e .
\]
\end{thm}
The set $E$ can be chosen such that every $e \in E$ has a component which is non-zero in $e$,
but zero in all other elements of $E$.
Hence, $|E| \le \dim(S)$ and $|E| \le |\supp(f)|$.

It is easy to see that every flux mode is the conformal sum of EFMs.
For later use, we present a slightly rephrased version of this result.
We note that, if $e$ is an EFM,
then any element of the ray
\[
\{ \lambda \, e \,|\, \lambda>0 \}
\]
is an EFM.
Hence, we may refer to one representative EFM on each ray.

\begin{cor} \label{cor:tau}
Let $(\mathcal{S},\mathcal{R},N)$ be a metabolic network,
$\tau \in \{-,0,+\}^\mathcal{R}$ be a sign vector
and $E_\tau$ be a set of representative EFMs conforming to $\tau$.
Then, every flux mode $f \in \FC$ conforming to $\tau$
is a non-negative linear combination of elements of $E_\tau$:
\[
f = \sum_{e \in E_\tau} \alpha_e \, e \quad \text{ with } \alpha_e \ge 0 .
\]
\end{cor}
\begin{proof}
Clearly, $f \in \FC$ implies $f \in \ker N$.
By Theorem \ref{thm:conform}, $f$ is the conformal sum of EVs of $\ker N$.
However, for an EV $e \in \ker N$ to conform to $f \in \FC$, it is required that $e \in \FC$.
Hence, by Lemma \ref{lem:efmev}, $e$ is an EFM,
which can be written as a positive scalar multiple of a representative EFM.
\end{proof}

\subsection{Problem \col{reformulation}}

\col{
We start with the formal statement of an intuitive argument.
Consider a feasible solution of the enzyme allocation problem~\eqref{original}
and the corresponding steady-state flux:
if a reaction does not carry any flux,
then the corresponding optimal enzyme concentration is zero.
We add an appropriate constraint to the enzyme allocation problem
and obtain an equivalent optimization problem.
}%

\begin{lem} \label{lem:equivalent1}
\col{
Let $(\mathcal{S},\mathcal{R},N,v)$ be a kinetic metabolic network.
The enzyme allocation problem~\eqref{original}
is equivalent to the following optimization problem:
\begin{subequations} \label{original_add_con}
\begin{equation}
\max_{\bar{x} \in X, \, c \in \RR^\mathcal{R}_\ge} \bar{v}_\ro
\end{equation}
subject to
\begin{gather}
c_\rr = 0 \quad \text{if } \, \bar{v}_\rr = 0 , \\
\sum_{\rr \in \mathcal{R}} w_\rr \, c_\rr = c^\tot , \\
\bar{v} \in \FC , \, \bar{v}_\ro > 0 .
\end{gather}
\end{subequations}
}%
\end{lem}

\begin{proof}
\col{
For every feasible solution $(\bar{x},c)$ of \eqref{original}
with objective function $\bar{v}_\ro$,
we construct a feasible solution $(\bar{x},c')$ of \eqref{original_add_con}
with objective function $\bar{v}'_\ro \ge \bar{v}_\ro$:
Let $S_{\bar{v}} = \supp(\bar{v})$.
Using $\lambda = c^\tot / ( \sum_{\rr \in S_{\bar{v}}} w_\rr \, c_\rr ) \ge 1$,
we set
\[
c'_\rr =
\begin{cases}
\lambda \, c_\rr & \text{if } \rr \in S_{\bar{v}} , \\
0 & \text{if } \rr \notin S_{\bar{v}} .
\end{cases}
\]
}%
\col{
Clearly,
\[
\sum_{\rr \in \mathcal{R}} w_\rr \, c'_\rr
= \sum_{\rr \in S_{\bar{v}}} w_\rr \, c'_\rr
= \lambda \sum_{\rr \in S_{\bar{v}}} w_\rr \, c_\rr
= c^\tot .
\]
Further, $\bar{v}' = c' \mal \ka = \lambda \, (c \mal \ka) = \lambda \, \bar{v}$
implies $N \bar{v}' = 0$ and $\sigma(\bar{v}') = \sigma(\bar{v})$, that is, $\bar{v}' \in \FC$.
Hence, $(\bar{x},c')$ fulfills constraints (\ref{original_add_con}bcd)
and $\bar{v}'_\ro = \lambda \, \bar{v}_\ro \ge \bar{v}_\ro$.
}%
\end{proof}

\col{As a consequence,}
variation over enzyme concentrations can be replaced by variation over flux modes.

\begin{lem} \label{lem:equivalent2}
\col{
Let $(\mathcal{S},\mathcal{R},N,v)$ be a kinetic metabolic network.
The enzyme allocation problem~\eqref{original_add_con}
is equivalent to the following optimization problem over $\bar{x} \in X$ and $f \in \FC$:
\begin{subequations} \label{reformulated}
\begin{equation}
\min_{\bar{x} \in X, \, f \in C} \sum_{\rr \in \supp(\ka)} \frac{w_\rr f_\rr}{\ka_\rr(\bar{x},p)}
\end{equation}
subject to
\begin{gather}
\sigma(f) \le \sigma(\ka), \, f_\ro = 1 .
\end{gather}
\end{subequations}
Let $(\bar{x},c)$ and $(\bar{x},f)$ be corresponding feasible solutions of \eqref{original_add_con} and \eqref{reformulated}, respectively.
The product of the related objective functions amounts to
\begin{equation} \label{obj_fun}
\bar{v}_\ro \sum_{\rr \in \supp(\ka)} \frac{w_\rr f_\rr}{\ka_\rr} = c^\tot .
\end{equation}
}%
\end{lem}

\begin{proof}
\col{
We show that,
for every feasible solution $(\bar{x},c)$ of \eqref{original_add_con},
there exists a feasible solution $(\bar{x},f)$ of \eqref{reformulated}, and vice versa.
Moreover, that the related objective functions fulfill Equation~\eqref{obj_fun}.
}%

\col{
Assume that $(\bar{x},c)$ is a feasible solution of \eqref{original_add_con}.
Define $f = \bar{v} / \bar{v}_\ro$.
Clearly, $f \in \FC$,  $f_\ro = 1$,
and $\sigma(f) = \sigma(\bar{v}) = \sigma(c \circ \ka) \le \sigma(\ka)$.
Hence,
$(\bar{x},f)$ is a feasible solution of \eqref{reformulated}.
Let $S_\ka = \supp(\ka)$.
Using $\bar{v} = c \circ \ka = \bar{v}_\ro f$ and hence $c_\rr = \bar{v}_\ro f_\rr / \ka_\rr$ for $\rr \in S_\ka$,
we rewrite the sum constraint and obtain the desired Equation~\eqref{obj_fun}:
\[
c^\tot
= \sum_{\rr \in \mathcal{R}} w_\rr \, c_\rr
= \sum_{\rr \in S_\ka} w_\rr \, c_\rr
= \bar{v}_\ro \sum_{\rr \in S_\ka} \frac{w_\rr f_\rr}{\ka_\rr} .
\]
}%

\col{
Conversely, assume that $(\bar{x},f)$ is a feasible solution of \eqref{reformulated}.
Since $\sigma(f) \le \sigma(\ka)$,
we can define $\bar{v}_\ro > 0$ by Equation~\eqref{obj_fun},
and we set $c_\rr = \bar{v}_\ro f_\rr / \ka_\rr$ for $r \in S_\ka$
and $c_\rr = 0$ for $r \notin S_\ka$.
Clearly, $c_\rr = 0$ if $\bar{v}_\rr = c_\rr \, \ka_\rr = 0$.
Further,
\[
\sum_{\rr \in \mathcal{R}} w_\rr \, c_\rr
= \sum_{\rr \in S_\ka} w_\rr \, c_\rr
= \bar{v}_\ro \sum_{\rr \in S_\ka} \frac{w_\rr f_\rr}{\ka_\rr}
= c^\tot .
\]
By definition,
$\bar{v}_\rr = c_\rr \, \ka_\rr = \bar{v}_\ro f_\rr$ for $r \in S_\ka$,
and, since $\sigma(f) \le \sigma(\ka)$,
$\bar{v}_\rr = c_\rr \, \ka_\rr = 0$ and $f_\rr = 0$ for $r \notin S_\ka$.
That is, $\bar{v} = \bar{v}_\ro f \in \FC$.
Hence,
$(\bar{x},c)$ is a feasible solution of \eqref{original_add_con}.
}%
\end{proof}

We note that the inequality constraints involving the kinetics may be unfeasible.
\col{For given flux mode $f \in \FC$,}
the existence of steady-state metabolite concentrations $\bar{x} \in \RR^\mathcal{S}_\ge$
such that $\col{\sigma(f)} \le \sigma(\ka(\bar{x},p))$
is equivalent to the existence of chemical potentials $\mu \in \RR^\mathcal{S}$ such that $\col{\sigma(f)} \le \sigma(- (\mu N)^T)$.
Whereas conventional FBA has to be augmented with thermodynamic constraints~\cite{BeardBabsonCurtisQian2004,MuellerBockmayr2013},
they are incorporated in the definition of a metabolic network \col{with known kinetics}.

\subsection{\col{Main results}}

The next statement characterizes \col{optimal} solutions of the enzyme allocation problem
\col{for fixed metabolite concentrations}.
Its proof involves the result on conformal sums obtained in Subsection~\ref{subsec:sign}.

\begin{pro} \label{pro:main1}
Let $(\mathcal{S},\mathcal{R},N,v)$ be a kinetic metabolic network.
\col{Consider} the enzyme allocation problem~\eqref{original} \col{for fixed $\bar{x} \in X$.}
\col{If this restricted optimization problem is feasible,}
then it has an optimal solution
for which the corresponding steady-state flux is an EFM.
\end{pro}

\begin{proof}
\col{
By Lemmas~\ref{lem:equivalent1} and \ref{lem:equivalent2},
the enzyme allocation problem~\eqref{original} is equivalent to optimization problem~\eqref{reformulated}.
We consider \eqref{reformulated} for fixed $\bar{x} \in X$ and assume that this restricted problem is feasible.
}%

\col{
We write $\ka$ short for $\ka(\bar{x},p)$
and introduce $\tau = \sigma(\ka)$ and $S_\ka = \supp(\ka)$.
In \eqref{reformulated},
we vary over $f \in \FC$ such that $\sigma(f) \le \sigma(\ka) = \tau$ and $f_\ro = 1$.
}%
By Corollary \ref{cor:tau},
every flux mode $f \in \FC$ conforming to $\tau$ is a non-negative linear combination of elements of $E_\tau$,
which is a set of representative EFMs conforming to $\tau$.
We assume the EFMs to be scaled by component $\ro$
and divide the set $E_\tau$ into two subsets, $E_\tau = E_1 \cup E_0$,
such that $e \in E_1$ implies $e_\ro = 1$ and $e \in E_0$ implies $e_\ro = 0$.
We have:
\[
f = \sum_{e \in E_1} \alpha_e \, e
  + \sum_{e \in E_0} \beta_e \, e
\quad \text{ with }  \alpha_e, \beta_e \ge 0 .
\]
From $f_\ro = 1$, we obtain the constraint
\[
1 = f_\ro
= \sum_{e \in E_1} \alpha_e \, e_\ro
+ \sum_{e \in E_0} \beta_e  \, e_\ro
= \sum_{e \in E_1} \alpha_e .
\]

Using the conformal sum for $f$ in \eqref{reformulated},
we obtain an equivalent formulation of the restricted problem:
\begin{subequations}
\begin{equation}
\min_{\alpha_e,\beta_e} \sum_{\rr \in S_\ka}
\frac{w_\rr \left( \sum_{e \in E_1} \alpha_e \, e_\rr
    + \sum_{e \in E_0} \beta_e  \, e_\rr \right)}{\ka_\rr}
\end{equation}
subject to
\begin{gather}
\sum_{e \in E_1} \alpha_e = 1 .
\end{gather}
\end{subequations}

We observe that the objective function is linear in $\alpha_e$ and $\beta_e$:
\[
  \sum_{e \in E_1} \alpha_e \underbrace{ \sum_{\rr \in S_\ka} \frac{w_\rr \, e_\rr}{\ka_\rr} }_{g_e}
+ \sum_{e \in E_0} \beta_e  \underbrace{ \sum_{\rr \in S_\ka} \frac{w_\rr \, e_\rr}{\ka_\rr} }_{g_e}
\]
with
\[
g_e =
\sum_{\rr \in S_\ka} \frac{w_\rr \, e_\rr}{\ka_\rr} \quad \col{\text{for } e \in E_\tau} .
\]
Since all $e \in E_\tau$ conform to \col{$\tau = \sigma(\ka)$, that is}, $\sigma(e) \le \sigma(\ka)$,
we have $\frac{w_\rr \, e_\rr}{\ka_\rr} \ge 0$ for all $e \in E_\tau$ and $\rr \in S_\ka$. 
Moreover, for all $e \in E_\tau$, there is $\rr \in S_\ka$ such that $e_\rr \neq 0$ and hence $\frac{w_\rr \, e_\rr}{\ka_\rr} > 0$.
Consequently, $g_e>0$ for all $e \in E_\tau$.

Since there is no further restriction on $\beta_e \ge 0$,
the minimum of the objective function is attained at $\beta_e = 0$ for all $e \in E_0$.
In other words, EFMs $e \in E_0$ do not contribute to the optimal solution.

\col{Let $e' \in E_1$ be an EFM such that $g_{e'} \le g_e$ for all $e \in E_1$.}
Since $\alpha_e \ge 0$ and $\sum_{e \in E_1} \alpha_e = 1$, we have
\col{
\[
\sum_{e \in E_1} \alpha_e \, g_e \ge \sum_{e \in E_1} \alpha_e \, g_{e'} = g_{e'} ,
\]
}%
and the minimum of the objective function is attained at $\alpha_{e'} = 1$ and $\alpha_e = 0$ for all other $e \in E_1$,
that is, for $f=e'$.
\col{To conclude}, we consider a degenerate case:
If there are several $e \in E_{\min} \subseteq E_1$ for which $g_e$ is minimal,
then any $f = \sum_{e \in E_{\min}} \alpha_e \, e$
(with $\alpha_e\ge0 $ and $\sum_{e \in E_{\min}} \alpha_e = 1$) is optimal.
\end{proof}

The following statement is the main result of this work.

\begin{thm} \label{thm:main2}
Let $(\mathcal{S},\mathcal{R},N,v)$ be a kinetic metabolic network.
\col{
If the enzyme allocation problem~\eqref{original} has an optimal solution,}
then it has an optimal solution
for which the corresponding steady-state flux is an EFM.
\end{thm}

\begin{proof}
\col{
Let an optimal solution of the enzyme allocation problem~\eqref{original} be attained at $\bar{x} \in X$.
Clearly, optimization problem~\eqref{original} restricted to this particular $\bar{x}$ is feasible.
By Proposition~\ref{pro:main1},
this restricted problem has an optimal solution
for which the corresponding steady-state flux is an EFM.
}%
\end{proof}

In applications, we use Theorem~\ref{thm:main2} to study the switching behavior
of kinetic metabolic networks.
Depending on external parameters,
the optimal solution of the enzyme allocation problem may switch from one EFM to another,
involving a discontinuous change of enzyme and metabolite concentrations.
In a first approach,
one may vary the external parameters
and determine the optimal solution for each EFM
in order to find the optimal solution of the full problem.
To this end, we transform the optimization problem restricted to an EFM.

\begin{cor} \label{cor:efm}
Let $(\mathcal{S},\mathcal{R},N,v)$ be a kinetic metabolic network.
In the enzyme allocation problem~\eqref{original},
let the steady-state flux be restricted to $\{ \lambda \, e \,|\, \lambda>0 \}$,
where $e \in \FC$ is an EFM \col{with $e_\ro = 1$}.
Then,
\col{this} restricted optimization problem is equivalent to \col{the following} optimization problem over $\bar{x} \in X$:
\begin{subequations} \label{reformulated_efm}
\begin{equation}
\min_{\bar{x} \in X} \sum_{\rr \in \supp(e)} \frac{w_\rr \, e_\rr}{\ka_\rr(\bar{x},p)}
\end{equation}
subject to
\begin{gather}
\col{\sigma(e) \le \sigma(\ka) .}
\end{gather}
\end{subequations}
The corresponding enzyme concentrations $c \in \RR^\mathcal{R}_\ge$ \col{are given by}
\begin{equation}
c_\rr = c^\tot \, \frac{ \frac{e_\rr}{\ka_\rr(\bar{x},p)} }
{ \sum_{s \in \supp(e)} \frac{w_s \, e_s}{\ka_s(\bar{x},p)} } .
\end{equation}
\end{cor}

\begin{proof}
\col{
By Lemmas~\ref{lem:equivalent1} and \ref{lem:equivalent2},
the enzyme allocation problem~\eqref{original} is equivalent to optimization problem~\eqref{reformulated}.
Hence, we consider \eqref{reformulated} for fixed $f=e$.
Clearly, this restricted problem is feasible if and only if optimization problem~\eqref{reformulated_efm} is feasible.
If $\sigma(e) \le \sigma(\ka)$, then $\supp(e) \subseteq \supp(\ka)$ and hence
\[
\sum_{\rr \in \supp(\ka)} \frac{w_\rr \, e_\rr}{\ka_\rr(\bar{x},p)}
=
\sum_{\rr \in \supp(e)} \frac{w_\rr \, e_\rr}{\ka_\rr(\bar{x},p)} .
\]
That is, the objective functions of the two optimization problems are identical.
}%

\col{
By using $c \circ \ka = \lambda \, e$ \col{with $\lambda > 0$}
and the constraint $\sum_{\rr \in \mathcal{R}} w_\rr \, c_\rr = c^\tot$,
we obtain $\lambda = c^\tot / ( \sum_{\rr \in \supp(e)} \frac{w_\rr \, e_\rr}{\ka_\rr} )$
and hence $c_\rr = \lambda \, \frac{e_\rr}{\ka_\rr}$ for $\rr \in \mathcal{R}$.
}%
\end{proof}


\section{Discussion and outlook}

In this work, we have studied a general metabolic optimization problem under enzymatic \col{capacity} constraints.
\col{Our analysis was motivated by the fact} that the total enzymatic capacity of a metabolic
network is limited by finite shared resources,
such that an increase in the concentration of one enzyme
necessitates a decrease in the concentration of one or more other enzymes.
This scenario can be caused by molecular crowding, where enzymes compete for cytosolic space,
limited membrane space, a finite availability of \col{macro-nutrients, such as nitrogen or phosphorus, 
or micro-nutrients, such as transition metals,}
as well as a \col{limited} energy expenditure for amino-acid synthesis.
In each of these instances,
the global constraint can be formulated in terms of a weighted sum of enzyme concentrations,
where the weight factors specify the fraction of the shared resource utilized per unit enzyme. 

Recently, such global constraints have been incorporated into large-scale stoichiometric models of metabolism,
most notably to explain the occurrence of low-yield pathways~\cite{Schuster2011,Shlomi2011,Vazquez2011,Zhuang2011}. 
However, these works were restricted to linear stoichiometric optimization problems 
and did not consider enzyme concentrations within a kinetic description of a metabolic network. 
Here, we have addressed a non-linear kinetic optimization problem,
where we assume that each reaction can be catalyzed by an enzyme,
but allow \col{for} arbitrary enzyme kinetics.
Most importantly, we have derived a rigorous proof
that an optimal flux distribution under global enzymatic constraints is necessarily an elementary flux mode. 
This finding has significant consequences for our understanding of \col{metabolic optimality and} metabolic switches
as well as for the computational identification of optimal fluxes in kinetic metabolic networks. 

In particular, our results allow us to efficiently compute and compare the optimal enzyme distributions for individual flux modes.
In \col{a} recent work~\cite{Flamholz2013}, the trade-off between energy yield and protein cost
was studied for several alternative prokaryotic glycolytic pathways.
To this end, different pathway designs were compared
by fixing the output flux at a certain value and estimating the necessary protein investments. 
Clearly, such an approach does not exclude the possibility that a combined pathway might incur an even lower enzymatic cost.
However, our result shows that it is indeed sufficient to compare the solutions for EFMs.

Notwithstanding its theoretical merits, 
the practical implications of our approach have to be studied further.
Firstly, for genome-scale metabolic networks,
an exhaustive evaluation of all EFMs is computationally infeasible.
A natural next step is to design a combined linear/non-linear optimization algorithm
for the identification of optimal EFMs.
Secondly, while the enzymatic capacity constraint 
is applicable to a multitude of possible limitations,
our approach may be extended to include {\em de novo} synthesis or uptake of limiting resources.
We conjecture that the problem can still be formulated in a way that our overall conclusions remain valid
\col{if such more general models are considered}.
Thirdly, the effects of co-limitation in cellular metabolism
and the simultaneous optimization of more than one objective function
pose new challenges for theoretical analysis.

\col{Finally, we note that experimentally observed flux distributions are not necessarily always EFMs. 
While some well-known metabolic switches indeed show an exclusive choice between alternative metabolic states,
such as most instances of catabolite repression,
it is known that alternative metabolic strategies sometimes operate simultaneously,
such as a residual respiration in cancer cells~\cite{Moreno-Sanchez2007}. 
Whether such a co-occurrence of metabolic strategies arises due to additional constraints or optimality principles,
or to what extent these instances are sub-optimal adaptations, remains to be studied.
We note that mixed strategies were also observed in recent computational studies using large-scale FBA models~\cite{Vazquez2010}
as well as for integrated models of enzyme synthesis and metabolism
to investigate shifts in growth strategies~\cite{Molenaar2009}.
Again, our analytical results provide a strong incentive
to study the consequences of additional constraints or optimality principles}
on the solutions of metabolic optimization problems.
We believe that only the further understanding of the properties of optimal \col{flux distributions}
will allow us to investigate the fundamental trade-offs in cellular resource allocation.

%
%

\subsubsection*{Acknowledgments}

\col{
We acknowledge numerous helpful comments from two anonymous reviewers.
During the revision process, we were informed by the authors of \cite{Wortel2013}
that they have arrived at similar conclusions as presented in this work.
The results in the two papers were derived independently.
}%

\subsubsection*{Funding Statement}

This work is an output of the project
``Local Team and International Consortium for Computational Modeling of a Cyanobacterial Cell'',
Reg.~No.~CZ.1.07/ 2.3.00/20.0256, financed by the European Union and supported by the Ministry of Education of the Czech Republic (SM, RS),
as well as the project
``\"Ubergangsmetalle und phototrophes Wachstum: Ein neuer Ansatz der con\-straint-basierten Modellierung gro{\ss}er Stoffwechselnetzwerke''
funded by the Einstein Stiftung Berlin (RS).
The funding bodies had no role in study design, data analysis, and the decision to submit the manuscript.


\bibliographystyle{plain}
\bibliography{flux_opt,flux_opt_bio}


\end{document}